\documentclass[11pt]{article}
\usepackage{times}

\usepackage{amsmath}
\usepackage{amsthm}
\usepackage{amssymb}
\usepackage{color}
\usepackage{graphicx}
\usepackage{natbib}
\usepackage{url}

\usepackage{breakurl}
\usepackage[breaklinks]{hyperref}
\usepackage{fullpage}
\usepackage{fancyhdr}
\usepackage[font=small,labelfont=bf]{caption}

\newtheorem{theorem}{Theorem}[section]

\newtheorem{claim}[theorem]{Claim}

\newcommand{\posb}{\textsf{MoBS}}
\newcommand{\mobs}{\textsf{MoBS}}

\newcommand{\CV}{\textsf{CV}}

\newcommand{\BF}{\textsf{BF}}

\newcommand{\OR}{\textsf{OR}}
\newcommand{\unary}{\textsf{Unary Evaluation}}
\newcommand{\ue}{\textsf{UE}}

\newcommand{\binary}{\textsf{Binary Evaluation}}
\newcommand{\be}{\textsf{BE}}

\newcommand{\compare}{\textsf{Comparison}}
\newcommand{\sorting}{\textsf{Sorting}}
\newcommand{\cmos}{CMOS}
\newcommand{\dna}{DNA}

\newcommand{\prob}[1]{{\bf Pr}[#1]}
\newcommand{\var}[1]{{\bf Var}[#1]}
\newcommand{\E}{{\rm I\kern-.3em E}}

\newcommand{\blue}[1]{\textcolor{blue}{#1}}
\newcommand{\red}[1]{\textcolor{red}{#1}}
\newcommand{\magenta}[1]{\textcolor{magenta}{#1}}
\ifdefined\ShowComment
\newcommand{\comment}[1]{\noindent\textbf{\red{/* } }
	\textit{\blue{{#1}}}\textbf{\red{~*/}}
	\marginpar{$\leftarrow$\fbox{\red{ $\blacklozenge$}}} }%
\def\john#1{\marginpar{$\leftarrow$\fbox{J}}\footnote{$\Rightarrow$~{\sf #1 \magenta{--John}}}}
\def\krishna#1{\marginpar{$\leftarrow$\fbox{K}}\footnote{$\Rightarrow$~{\sf #1 \blue{--Krishna}}}}

\else
\newcommand{\comment}[1]{}
\def\john#1{}
\def\krishna#1{}

\fi

\title{Sustaining Moore's Law Through Inexactness \\ \large A Mathematical Foundation}

\author
{
John Augustine,$^{1\ast}$  Krishna Palem,$^{2}$ Parishkrati$^{1}$\\
\\
\normalsize{$^{1}$Department of Computer Science and Engineering, }\\
\normalsize{Indian Institute of Technology Madras, Chennai, India.}\\
\normalsize{$^{2}$Department of Computer Science, Rice University, Houston,
USA.}
\\~\\
\normalsize{$^\ast$Corresponding author; E-mail:  \url{augustine@iitm.ac.in}.}
}

\date{}

\begin{document}

\maketitle

\begin{abstract}
Inexact computing aims to compute good solutions that require considerably less resource --€" typically energy -- compared to computing exact solutions. While inexactness is motivated by concerns derived from technology scaling and Moore's law,  there is no formal or foundational framework for reasoning about this novel approach to designing algorithms. In this work, we present a fundamental relationship
between the quality of computing the value of a boolean function and the energy needed to compute it in a mathematically rigorous and general setting. On this basis, one can study the tradeoff between the quality of the solution to  a problem and the amount of
energy that is consumed. We accomplish this by introducing a computational model to classify problems based on notions of
symmetry inspired by physics.  We show that some problems are symmetric in that every input bit is, in
a sense, equally important, while other problems display a great deal of asymmetry in the importance
of input bits. We believe that our model is novel and provides a foundation for  inexact
Computing. Building on this, we show that asymmetric problems allow us to invest resources favoring the important bits -- a feature
that can be leveraged to design efficient inexact algorithms. On the negative side and in contrast, we can prove that the
best inexact algorithms for symmetric problems are no better than simply reducing the resource investment uniformly across all bits. Akin to classical theories concerned with space and time complexity, we believe the ability to classify problems as shown in our paper will serve as a basis for formally
reasoning about the effectiveness of inexactness in the context of a range of computational problems with energy being the primary resource.
\end{abstract}

Many believe that the exponential scaling afforded by Moore's law~\citep{moore} is reaching its limits
as transistors approach nanometer scales. Many of these limitations are based on physics based limits ranging
over thermodynamics and electromagnetic noise~\citep{kish} and optics~\citep{lithography}. Given that information technology is
the prime beneficiary of Moore's law, computers, memories, and related chip technologies are likely to be
affected the most. Given the tremendous value of sustaining Moore's law through information technology in a broad sense,
much effort has gone into sustaining Moore's law, notably through innovations in material science and electrical engineering.
Given the focus on information technology, a central tenet of these innovations has been to preserve
the behavior of \cmos~transistors and computing systems built from them.

While many of these innovations revolve around non-traditional materials such as graphene \citep{graphene1,graphene2}
supplementing or even replacing \cmos~\citep{graphene-computer}, exciting developments based on  alternate and
potentially radical models of computing have also emerged. Notable examples include \dna~\citep{dna1, dna2}, and quantum
computing frameworks~\citep{quantum1, quantum2, quantum3}. However, these exciting approaches and alternate models
face a common and potentially steep hurdle to becoming deployable technologies leading to  the preeminence of \cmos~as the material of choice. This brings the importance of Moore's law back to the fore and consequently,
in the foreseeable future,
the centrality of \cmos~to growth in information technologies remains.

The central theme of this paper is to develop a coherent theoretical foundation with the goal of reconciling these
competing concerns. On the one hand, continuing with \cmos~centric systems is widely believed to result in
hardware that is likely to be erroneous or function incorrectly in part. On the other hand, dating back
to the days of Alan Turing~\citeyearpar{turing} and explicitly tackled by von Neumann~\citeyearpar{vonNeumann}, a computer --- the
ubiquitous information technology vehicle --- has an unstated expectation that it has to function {\em correctly}.
This expectation of computers always functioning correctly as an essential feature
is at the very heart of our alarm about the doomsday scenario associated with the {\em end
of Moore's law}. For if one can use computers with faulty components as they are with concomitant but acceptable
errors in the computation, we could continue to use \cmos~transistors albeit functioning in a potentially unreliable regime.

Over the past decade,
this unorthodox approach to using a computer and related hardware such as memory
built out of faulty components, and used in this potentially faulty mode, referred to as {\em inexact computing},
has emerged as a viable alternative to coping with the Moore's law cliff.
\citet{PL13} and \citet{KP14} (and references therein) provide a reasonable overview of inexact computing practice.
 At its core, the counterintuitive
thesis behind inexactness is to note that, perhaps surprisingly, working with faulty components can in fact
result in computing systems that are {\em thermodynamically} more efficient~\citep{palem0,palem1, palem-cmos-japan}.
This approach simultaneously appeals to another hurdle facing the sustenance of Moore's law. Quite often referred to as the energy-wall or
power-wall, energy dissipation has reached such prohibitive levels that being able to cope with it is
the predominant concern in building
computer systems today. For example, to quote from an article from the New York Times~\citep{nyt} about the
potential afforded through inexact computing: ``If such
a computer were built in today's technologies, a so-called exascale computer would consume electricity equivalent
to 200,000 homes and might cost \$20 million or more to operate.''

While individual technological artifacts demonstrating the viability of inexact computing
might be many, a coherent understanding of how to design algorithms --- essential
to using inexact computing in large scale --- and understand the inherent limits
to the power of this idea  are not there. Such characterizations are typically the purview
of  theoretical computer
science, where   questions of designing efficient algorithms,  and inherent limits to
being able to design efficiently are studied.
While algorithm
design is concerned with finding efficient ways of solving problems, inherent
limits allow us to understand what is {\em not} possible under any circumstance within
the context of a mathematically well-defined model. Understanding what is inherent to computing in abstract terms has been a significant part of
these enquiries, and has evolved into the field referred to as computational complexity~\citep{AB09,MM11}.
 In this paper, we
present a computational complexity theoretic foundation to characterizing inexactness, and to the best of our knowledge for the
first time.

As in classical complexity theory, the atomic object at the heart of our foundation is a {\em bit}
of information. However, inexactness allows something entirely novel: {\em each bit}
is characterized by two attributes or dimensions, a {\em cost} and a {\em quality}.
Historically, physical {\em energy} was the cost and the {\em probability of correctness} was the quality~\citep{palem1,palem2}. As shown in Figure~\ref{fig:1} (originally reported in ~\citep{palem-cmos-japan}),
under this interpretation, a cost versus quality relationship was measured in the context of physically constructed \cmos~
gates. More recently, \citet{FKBSA15} have presented a voltage-scaled SRAM along with a characterization of the energy/error tradeoff,
where, unsurprisingly, we see that the bitcell error rate (BER) drops exponentially as  $V_{dd}$ increases.

\begin{figure}[htbp]
\begin{center}
    \includegraphics[scale=0.9,clip=true, trim=170 320 150 320]{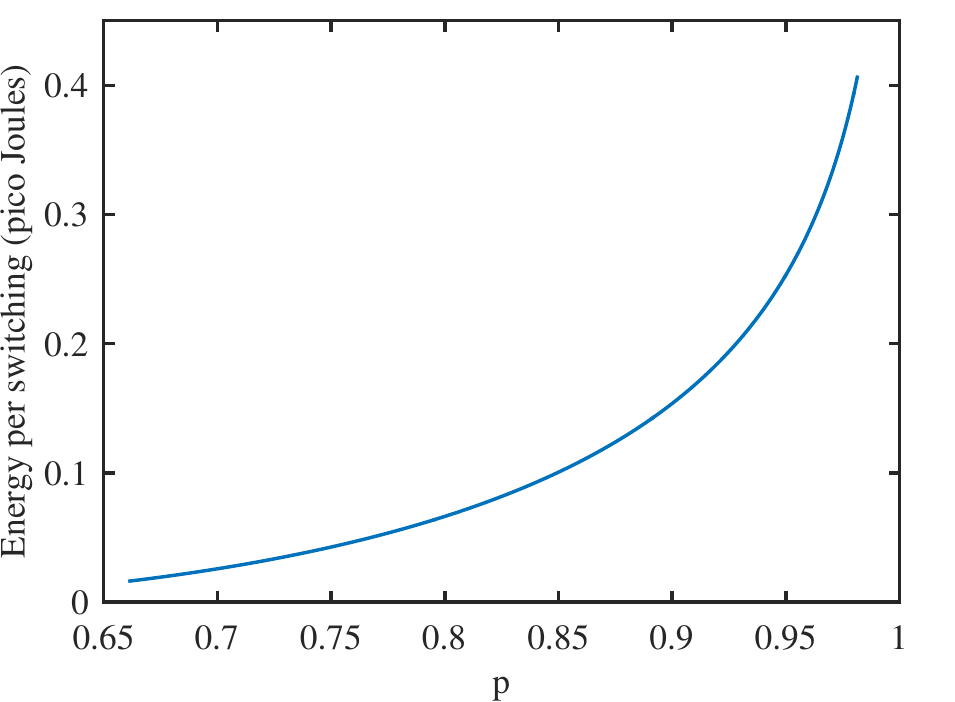}
\end{center}
\caption[Fig]{The quality-cost relationship referred to as the energy-probability relationship ($e$-$p$ relationship)
of a CMOS switch from~\citep{ACKP06,palem-cmos-japan}, where the energy spent $e$ increases exponentially with
 increasing probability of correctness $p$. Specifically, from~\citep{palem-cmos-japan}, the probability the operation is correct $p = 1- \frac{1}{2} {\sf erfc}\left (\frac{V_{dd}}{2 \sqrt{2} \sigma} \right )$, where $V_{dd}$ is the supply voltage, $\sigma$ is the standard deviation of additive gaussian noise, and the
complementary
error
function ${\sf erfc}(x) = \frac{2}{\sqrt{\pi}} \int_x^\infty e^{-u^2} du$.
}
 \label{fig:1}
\end{figure}

From a historical perspective, our work here builds naturally on the entire theme of
computing or reasoning in the presence of uncertainty~\citep{feige94, kenyon94}, which is
concerned about computing reliably in the presence of erroneous or uncertain information.
In this context, the bits in question have one of our two dimensions, quality alone. As a bit
is read, depending on the (unspecified) circumstance, it can be erroneous with a certain probability.
However, early work in inexactness (see Figure~\ref{fig:1}) showed that error or quality can in fact
be a parameter that can be related through the underlying physics (thermodynamics) to a cost, namely
energy. Thus, our work can be viewed as an extension of classical theoretical foundations to reasoning
about uncertainty, to one where we can trade uncertainty with cost: less uncertain being typically much
more expensive! In this sense, the degree of uncertainty of cost is an attribute that we can trade-off
based on what we wish to pay, as opposed to an externally imposed quantity that we are forced to
live with --- the theme of  prior work in this domain.

In order to have a clear basis for our discussion on computing, let us define a  {\em computational problem}  as evaluating a {\em Boolean
function} $f:\{0,1\}^n\rightarrow \mathbb{Z}$. Although quite elementary, such Boolean functions possess the ability to succinctly encode  input/output structure of problems in computing without compromising any of the mathematical rigor needed for careful analysis~\citep{boolean}.  Evaluating a Boolean function $f$ can be
represented as a {\em truth table} $T_f$ with $2^n$ rows corresponding
to each possible $n$-bit vector, and $n+1$ columns; we will use $c_f(i,j)$ (or just $c(i,j)$ when clear from context) to denote the $j$th element in the $i$th row of $T_f$. The first $n$
columns correspond to the $n$ input bit positions and the $(n+1)$th column
represents the output. To facilitate measuring the {\em quality} of the solution we produce, we view our output as a number in $\mathbb{Z}$. Since we are interested in inexact computing, let us suppose that an algorithm outputs $f'(I)$ for an input string $I \in \{0,1\}^n$, which could be different from the correct output $f(I)$. Then, the absolute difference $|f(I) - f'(I)|$ captures the magnitude  of error.  Although algorithms we consider in this paper are deterministic, it must be noted that the magnitude of error will be a random variable because input bits can be read incorrectly with some probability.

Let us consider some elementary examples of Boolean functions. The {\sf OR} problem for instance takes $n$ bits as input and outputs a 1 except when all input bits are zeros. The \unary~problem (or \ue ~in short) outputs the number of 1's in the input, while \binary~problem (or \be~in short) outputs $\sum_{i=0}^{n-1}2^i b_i$, where $(b_{n-1}, b_{n-2}, \ldots, b_0)$ is the input bit string.  These problems can be succinctly captured  in truth table representation as show in Table~\ref{tab:truth}.
\begin{table}[htp]
\caption{Truth table indicating the input and output for ${\sf OR}$, \ue, and \be.}
\begin{center}
\begin{tabular}{|c|c|c||c|c|c|} \hline
$b_2$ & $b_1$ & $b_0$ & {\sf OR} & \ue & \be \\ \hline \hline
0 & 0 & 0 & 0 & 0 & 0 \\ \hline
0 & 0 & 1 & 1 & 1 & 1 \\ \hline
0 & 1 & 0 & 1 & 1 & 2 \\ \hline
0 & 1 & 1 & 1 & 2 & 3 \\ \hline
1 & 0 & 0 & 1 & 1 & 4 \\ \hline
1 & 0 & 1 & 1 & 2 & 5 \\ \hline
1 & 1 & 0 & 1 & 2 & 6 \\ \hline
1 & 1 & 1 & 1 & 3 & 7 \\ \hline
\end{tabular}
\end{center}
\label{tab:truth}
\end{table}%

Let us now impose the two dimensions of cost and quality upon computational problems. Let us suppose that we invest a cost (or energy) $e_i$ on each bit $b_i$. Under finite values of $e_i$, there will be a loss in quality of the bits. Therefore, our view of the input bits will be restricted to an approximate bit vector $(b'_0, b'_1, \ldots, b'_n)$. We then (inspired by \citep{ACKP06}) model the quality in probabilistic terms as:
\[
\prob{b_i \ne b'_i} = 2^{-e_i}.
\]
Thus, under a finite energy budget, the outcome will be an approximation with the quality of the approximation increasing as the energy budget is increased.

Let us now consider a thought experiment  that   brings out two different modes in which algorithms can operate via a novel and elegant use of permutation groups. Our algorithm must assign each input bit with an energy value such that the total energy is within some finite budget. Suppose an adversary has the power to permute the energy values (but not the associated bits) by choosing a permutation $\pi$ uniformly at random from some permutation group $G$ (with the two extremes, the identity permutation $I_n$ and the symmetric group $S_n$, being of most interest). This will imply that the energy values associated with each bit could change. The algorithm is aware of $G$, but not the exact permutation $\pi$ chosen by the adversary. The question that comes to mind now is: how much can the adversary affect the outcome? Clearly, a well-designed optimal algorithm will try to compensate for this adversarial intrusion in some way, but can it succeed in mitigating the effects of this adversarial permutation?

Let us consider the {\sf OR} problem. If any one of the input bits is a 1, the correct output is a 1, and in this sense, every bit has equal impact on the output. Therefore, intuitively, an optimal algorithm under a finite energy budget must allocate equal energy to every bit. (We will formally prove this shortly.) When the energy values across the input bits are the same, the adversary is rendered toothless, and cannot impact the quality of the outcome in anyway. Therefore, we say that the {\sf OR} problem is {\em symmetric} under $G$. In fact, it is straightforward to see that {\sf OR} is symmetric even under the symmetric group $S_n$ that captures all $n!$ permutations. With a little thought, one can also surmise that the \unary~problem is also symmetric under $S_n$.

To build a contrast, let us consider the \binary~problem. When the adversary is restricted to $I_n$, an optimal algorithm will assign more energy to bit $b_{n-1}$ as its ``impact" on the output can be as high as $2^{n-1}$. However, when the adversary is empowered with $S_n$, the algorithm cannot favor bit $b_{n-1}$ over $b_0$. Therefore, the adversary can significantly affect the quality of the output. Thus, \binary\ is said to be an {\em asymmetric} problem.

There is a curious  perhaps even striking connection between the formulation above and
the role that symmetries and asymmetries played in physics at the turn of the last century.
Perhaps the earliest work that historians point to as a basis of this connection is
the work of Pierre Curie more than a century ago widely recognized as  {\em Curie's
principle}~\citep{curie}. Informally speaking, he ties the appearance of phenomena to when
a system is ``transformed.''
 We can interpret this to mean that for a change of some type to occur in the system, we need an
absence of symmetries. Physicists tend to think of outcomes or effects in terms
of phenomena and so another way of interpreting Curie's powerful concept is to
note that for a phenomenon to occur or exist, there must be inherent asymmetries
in the system. In our own case, the example of Boolean evaluation mentioned above
has a curious analogical connection since the quality of the output is the basis for observing
change and, as noted above, asymmetries are an essential part of being able to observe change---since
symmetric case of the \OR~function for example will not exhibit any change. Thus, in our case also,
the existence of asymmetries is inherently necessary to observe changes in the ``quality'' of what
is being computed. We note in passing that Curie's work is the first that we are aware of which
formally, in a mathematical sense, captures symmetries and asymmetries using a group theoretic formulation.
This approach is also reflected in our own work where we use permutation groups as a basis for characterizing
symmetries and asymmetries, as outlined above through the three examples. We remark in passing that
several conditions must hold for Curie's principle to be applicable, which physicists have documented extensively,
and our analogical remark is a substantial simplification of the concept.

Following Curie, several historical figures in physics pursued the use of group theoretic symmetries and asymmetries
as formal tools in physics with perhaps another analogical connection to our symmetric case---the OR problem
described above. Best characterized in the work of Emmy Noether~\citeyearpar{noether} and again based on a group
theoretic foundation---Lie groups to be specific---paraphrased, Noether's theorem states that a physical system
that embodies symmetries will result in (physical) quantities to be preserved or conserved under transformations.
Thus, changes in systems which embody such symmetries will not yield observable changes in physical quantities
of interest such as momentum and energy. In our own symmetric case, for example in the \OR~problem,  we can
interpret the quality of the output to be preserved under the (permutation) transformation of the energy vector
stated above. Consequently, we can conclude that for symmetric functions in our sense, the quality is conserved
under such transformations. We wish to add that our own framework in this paper was inspired by the style of
thinking central to symmetries and symmetry breaking in physics as exemplified by the two cases discussed above,
but that the connection is analogical---we do not wish to imply any novel insights in the physics domain based on our work
presented in the sequel.

While we illustrated symmetry and asymmetry at the two extremes ({\sf OR} and \ue~at one end and \be~at the other), we can clearly envision a host of intermediate problems for which the symmetry is broken at various levels. We capture the level to which the symmetry can be broken by a parameter called the {\em measure of broken symmetry} or \mobs, which we formally define shortly. We show that the \mobs~for symmetric problems like {\sf OR} and \ue~is 1, but exponential in $n$ for asymmetric problems like \be. (See Table~\ref{tab:results} for a complete listing of results.)

\begin{table}[htp]

\begin{center}
\begin{tabular}{|c|c|c|}
\hline
{\bf Problem Name}  & $\posb$  \\ \hline
\OR~Problem  & 1 \\ \hline
$\unary$  & 1\\ \hline
$\binary$  & $2^{\Omega(n)})$ \\ \hline
$\compare$ two $k$-bit numbers  & $2^{\Omega(k)}$ \\ \hline
$\sorting$ $k$-bit numbers & $2^{\Omega(k)}$ \\ \hline
\end{tabular}
\end{center}
\caption{The $\posb$ of problems showing dramatic difference under asymmetric situations.}
\label{tab:results}
\end{table}%

\section{Computational Model and a Related Property}

Going beyond traditional notions of algorithms designed for exact inputs, in our model of computation, we empower our algorithms with the ability (within well-defined bounds) to statistically alter  noise characteristics in the input data {\em and}  adapt its computational steps accordingly. We allow an algorithm $A$ to work within an energy budget $E$. Given $E$, the algorithm must specify an energy vector $E^A= (e_0, e_1, \ldots, e_{n-1})$ (with $\sum_j e_j \le E$) that will be used for reading the $n$ input bits. Moreover, the algorithm is also aware of the permutation group $P^A$ that will be employed for permuting the energy vector.
Let $\sigma$ be a permutation
drawn uniformly at random from $P^A$. Let us denote the energy vector $E^A$
permuted by $\sigma$ as $\vec{E}^A_\sigma = (e_{\sigma(0)},
e_{\sigma(1)}, \ldots, e_{\sigma(n-1)})$.
Algorithm $A$ does not have direct
access to the input row $i$ in the truth table, but rather receives the input row after
 the energy vector $E^A_\sigma$ has been applied to $i$. More precisely, each cell $c(i, j)$ is read correctly with
probability $1-2^{-e_{\sigma(j)}}$ and  incorrectly with probability $2^{-e_{\sigma(j)}}$.
 Note that while $A$ is aware of the permutation group $P^A$, the exact permutation $\sigma$  is not revealed to $A$.
When $P^A$ contains only the identity permutation, we say that $A$  is {\em clairvoyant} because the adversary is incapable of hiding or altering the association between energy values and bit positions.
Otherwise, we  say that $A$ is {\em blindfolded by $P^A$}.
When $P^A = S_n$, the symmetric group defined on all $n!$ permutations, we simply say that $A$ is {\em blindfolded}.  One may astutely wonder if the use of the permutation group is somehow restricting the adversary. We emphasize that this is not the case. We will go on to show that the best algorithm in the blindfolded setting elegantly corresponds  to the traditional computational model in which equal energy is invested in each bit.

Let us now consider how algorithm $A$ can compute $f$.
We interpret the
behavior of $A$ as an attempt to decipher either the input row $i$ that was fed in as
the input or another row that produces the same output.
Let us suppose that $A$ concludes that the input row is $A(i)$.
We restrict $A$ to behave deterministically even though the input row
that it actually sees is random due to the application of $E^A_\sigma$.
Therefore, we can view $c(A(i), n)$ as a random variable over the set of all
legal output values, i.e., all values in the $n$th column of $T_f$, but the
randomness is over the probability space induced by $E^A$ and $P^A$.
Thus, the worst case probability that $A$ is incorrect  is
$\max_i \prob{c(A(i), n) \ne c(i,n) }$ and the {\em quality of the algorithm} $Q(A)
$   (which in this case can be interpreted as the  expected number of
correct executions in the worst case taken over all input rows before we
see an incorrect execution) is simply $\min_i (1/\prob{c(A(i), n) \ne
c(i,n) })$.

 In the clairvoyant setting, each bit $j$ is read incorrectly with probability $2^{-
e_j}$ and therefore, a clairvoyant algorithm can assign energy in
proportion to the correctness required for bit $j$.
Let us now consider the consequence of blindfolding an algorithm.
\begin{claim}[Blindfolding Claim] \label{claim:blind}
When we blindfold an
algorithm,  the probability that any bit is read
incorrectly is at least $2^{-E/n}$, where $E = \sum_j e_j$.
\end{claim}
\begin{proof} As a result of the blindfolding, each of the $n$ entries in the energy vector is
applied with equal probability and therefore, $\prob{\text{$j$th bit is read
incorrectly}} = \sum_j (1/n)2^{-e_{\sigma(j)}}$, where $\sigma$ is a random permutation.
Since the arithmetic mean is at least  the
geometric mean, $\sum_j (1/n)2^{-e_{\sigma(j)}} \ge \left( \Pi_{j=0}^{n-1} 2^{-e_{\sigma(j)}} \right)^{1/n}
= 2^{-E/n}$ as claimed.
\end{proof}
The implication of this claim is that the probability of error in each bit
cannot be improved by employing an energy vector with non-uniform
energy entries. This means that the best algorithm under the blindfolded setting corresponds to the current
trend in computing whereby equal effort or energy is expended in reading each bit,
but the clairvoyant setting opens up the possibility for variations.
 To capture a sense of the price we pay when employing uniform energy
vectors, we define the {\em measure of broken symmetry} (MoBS) for computing $f$ as
\[
\posb(f) =\max_E \max_i \frac{\prob{c(\BF_E(i), n) \ne c(i,n) }}{\prob{c(\CV_E(i), n) \ne c(i,n)} },
\] where $\CV_E$ and $\BF_E$ are the  optimal clairvoyant and optimal
blindfolded algorithms, resp., while both are restricted to an energy budget of $E$.
 When $\mobs(f)=1$ for the problem of computing some function $f$, then, we can infer that computing $f$ is identical under both the clairvoyant and the blindfolded setting. This reveals an inherent symmetry in the problem.
\section{Symmetric Problems}

To see how the notion of $\posb$ helps us understand the
applicability of non-uniform energy distributions, let us consider two
fundamental canonical problems. Consider first the $\unary$ problem (in
short $\ue$ problem) that requires us to count the number $k$ of bits set
to 1.   Consider  the optimal clairvoyant algorithm $\CV$ and let us say
that it employs energy vector  $E_{\CV}= (e_0, e_1, \ldots, e_{n-1})$.
Let $(b_0, b_1, \ldots, b_{n-1})$ be the input vector and $(b'_0, b'_1, \ldots, b'_{n-1})$
be the vector read by the algorithm $\CV$ under the energy vector $E_{\CV}$.

We claim that $E_{\CV}$ is in fact the uniform vector in which each
$e_j = E/n$, where $E = \sum_j e_j$ is the total energy budget. For
the sake of contradiction, let us assume  an arbitrary $E_{\CV}$.
Consider bit $j$. Since $\prob{b_j \ne b'_j} = 2^{-e_j}$, we can extend
this to stating that $\prob{b_j = 1} = b'_j (1- 2^{-e_j}) + (1-b'_j) 2^{-e_j}$.
Define Poisson trial variable $X_i$ that takes the value 1 with probability
$ b'_j (1- 2^{-e_j}) + (1-b'_j) 2^{-e_j}$ and 0 otherwise. Since each bit $j$ is  read
independently of other bits, the algorithm can be viewed as optimizing
how well $X = \sum_j X_j$ estimates $\sum_j b_j$. Since $X$ is
the sum of $n$ Poisson trials, we can get the optimal estimation
by minimizing the variance of $X$.  Notice that $\var{X_j} = (1-2^{-e_j})2^{-e_j}$.
Notice that for any two positions $j$ and $j'$,
their combined variance $\var{X_j}+\var{X'_j}  = (1-2^{-e_j})2^{-e_j}+(1-2^{-e_j'})2^{-e_j'}$
is minimized when $e_j = e_j'$ when $e_j+e_j'$ is fixed. This can be easily
extended to computing $\var{X}$ where we can show that the variance
is minimized when $e_j = E/n$ for every $j$. We can therefore claim
that $\CV$ uses a uniform energy vector.
Thus, $\posb(\unary) = 1$ because the same optimal algorithm can be employed in the blindfolded setting.

More generally,  $\mobs(f)=1$ for all symmetric Boolean functions $f$ where the outcome $f(I)$ equals $f(\sigma(I))$, where $\sigma$ is any permutation of the input bits $I$~\citep{boolean}. The {\sf OR} problem discussed earlier is a canonical example. But the set of symmetric problems is larger than the the set of symmetric Boolean functions. Consider for example, the {\sf Tribes} problem where (in a simplified sense), the $n$ input bits are partitioned into two tribes: the first $n/2$ bits and the second $n/2$ bits. The output of the {\sf Tribes} function is a one iff at least one of the tribes consists of all 1 bits. This function is not a symmetric Boolean function as is witnessed by {\sf Tribes}(0011)=1 while {\sf Tribes}(0101)=0. However, the problem of evaluating the {\sf Tribes} function is a symmetric problem.

\begin{figure}[htbp]
\begin{center}
\includegraphics[scale=0.45]{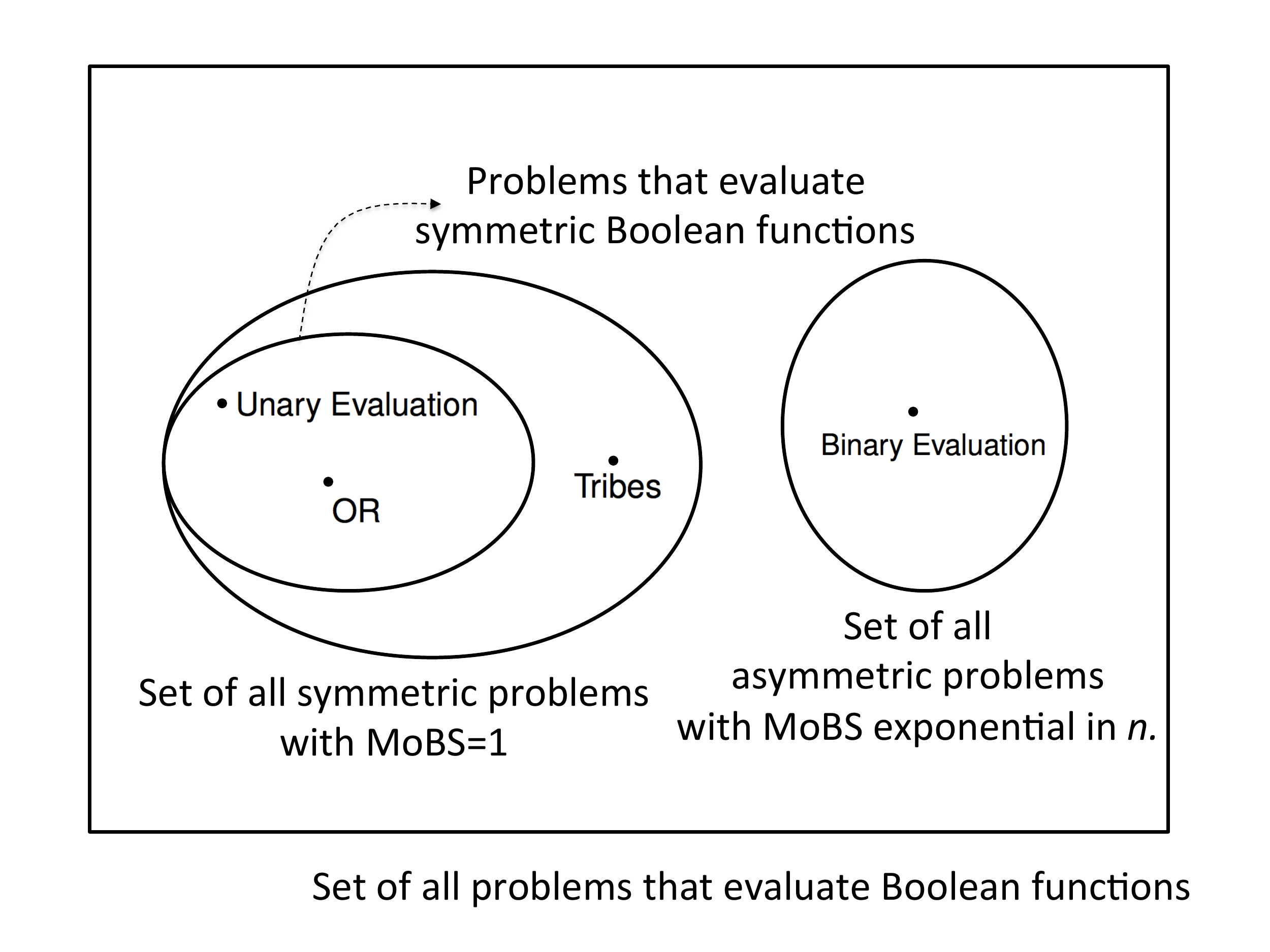}
\caption{A Venn diagram showing the various classes of computational problems that evaluate Boolean functions.}
\label{default}
\end{center}
\end{figure}

\section{Asymmetric Problems}
Thankfully, most real-world problems have quite the opposite structure
wherein there is significant benefit to employing non-uniform energy
vectors. To illustrate this, consider $\binary$ (or $\be$ for short) where
we have to compute $\sum_{j=0}^{n-1} 2^j b_j$, which arguably is the
most fundamental problem. Clearly, the $(n-1)$th bit has impact $2^j$,
which is significantly more  important than the $0$th bit and
therefore displays a marked difference from the $\unary$. One can
intuitively see that non-uniform energy distribution ought to provide
significant improvement.

In order to formalize this intuition about $\binary$,
consider an energy budget of $E=n(n+1)/2$. The blindfolded algorithm will use $e=(n+1)/2$ units of energy for each bit.
The behavior of the blindfolded algorithm can be viewed as essentially evaluating $\sum_j w_j (b_j'(1-2^{-e}) +  (1-b_j')2^{-e})$. Without loss in generality, if we assume that bit $b_{n-1}'=0$, then, the expected evaluated quantity --- taking just the most significant bit into account and also noting that $w_{n-1}=2^{n-1}$ --- is at least $2^{n-1}\cdot 2^{-(n+1)/2} = 2^{(n-3)/2}$.
However, in the clairvoyant setting, we can set $e_j = j+1$, which respects the energy budget constraint. The expected error is $\sum_j  w_j 2^{-j-1} = n/2$. Under characteristic quality function defined as the reciprocal of the expected error, therefore, $\posb(\be) = \Omega(2^{n/2}/n)$.  (This analysis is based on the unpublished work by \citet{CP11}.)

Let us similarly analyze $\compare$, the problem of comparing two $k= n/2$ bit numbers $x$ and $y$ to evaluate which one is larger. In the inexact computing perspective, it is more important to distinguish the two numbers when they are far apart, i.e., when $|x-y|$ is large, than when they are very close to each other.  With that perspective in mind, we define the quality function to be $1/(|x-y| \cdot \prob{\text{$x$ and $y$ are incorrectly compared}})$. Let us consider an energy budget of $k(k+1)/2$. Under the blindfolded setting, each bit-wise comparison will get $(k+1)/2$ units of energy, but in the quality optimal setting, we can assign $j+1$ units of energy for comparing $j$th bit positions, $0 \le j \le k-1$. Extending the argument from $\be$, we will get $\posb(\compare) = 2^{\Omega(n)}$.

Let us now consider the $\sorting$ problem that takes $L$ numbers $(x_1, x_2, \ldots, x_L)$, each $k$ bits long, and reorders them in non-decreasing order. This will make the total number of input bits $n = kL$.  Consider a pair of numbers $x_{\ell_1} $ and $x_{\ell_2}$, $\ell_1 \ne \ell_2$.  As in the case of $\compare$, we can tolerate the two numbers being wrongly ordered in the output sequence when $|x_{\ell_1}  - x_{\ell_2}|$ is small, but not when the difference is large.  Thus, as a natural extension to the quality function defined for $\compare$, we consider the following as the quality function for $\sorting$ wherein the probability of error in the ordering of a pair of numbers is weighted by their magnitude difference:
\begin{equation}
Q_\sorting = \frac{1}{\sum_{\ell_1 < \ell_2} \left(|x_{\ell_1} - x_{\ell_2}| \cdot \prob{\text{$x_{\ell_1}$ and $x_{\ell_2}$ are wrongly ordered}}\right)}. \label{sorting}
\end{equation}
Consider the input in which $L/2$ of the numbers are of the binary form $(100\cdots 0)$ and the remaining are of the form $(000\cdots 0)$.
There are at least $L^2/4$ pairs that we call {\em expensive} pairs for which the absolute magnitude difference is $2^{k-1}$.  Let us now consider the expensive pairs under the blindfolded and the quality optimal scenarios. Extending the ideas from comparison, their probabilities of wrongly comparing an expensive pair under the blindfolded setting is at least $2^{-(k+1)/2}$. In the quality optimal setting, however, the probability will reduce to $2^{-k}$. Substituting these probability values into Equation \ref{sorting} and subsequently into the formula for $\posb$, we get:
\begin{equation}
\posb(\sorting) \ge \frac{\sum_{\text{$\ell_1$ and $\ell_2$ are expensive}} \left(|x_{\ell_1} - x_{\ell_2}| \cdot  2^{{-(k+1)/2}} \right)}{\sum_{\text{$\ell_1$ and $\ell_2$ are expensive}} \left(|x_{\ell_1} - x_{\ell_2}| \cdot  2^{-k}\right)} = 2^{\Omega(k)}.
\end{equation}

\section*{Remarks}
Our work is built on the complexity theoretic philosophy of understanding the inherent resource needs of problems under various computational models,  and classifying them based on those insights. From a utilitarian perspective, we now know that symmetric problems will not lead to significant gains from the principles of inexact computing. Thankfully, most real world problems display quite a bit of asymmetry where, we believe, opportunities abound, thus opening the door for future work in optimization techniques a practitioner may employ.  Moreover, hardware that is capable of trading error for energy is not just a theoretical possibility, but a reality with several groups spanning industry and academia actively pursuing their fabrication. In fact, our model with its emphasis on memory errors is in fact a reflection of current pursuits in memory technology~\citep{FKBSA15}.

Finally, for clarity, we have focussed on problems where the quality of bits are probabilistically  independent of each other, but this is somewhat simplistic as evidenced for example by the effect of carry propagating through  adders~\citep{C08,P16}. Arbitrary Boolean functions are likely to be riddled with much more complex dependencies and  we believe that there is significant scope for future work.

\section*{Acknowledgements}
We thank Dror Fried, Anshumali Shrivastava,  and Eli Upfal for providing us with useful comments. We also thank Jeremy Schlachter for using the analytical model from \citep{palem-cmos-japan} to recreate Figure~\ref{fig:1}.

\small

\bibliographystyle{plainnat}

\end{document}